\documentclass[runningheads]{llncs}
\usepackage[margin=2cm]{geometry}

\usepackage[utf8]{inputenc}
\usepackage[T1]{fontenc}
\usepackage{amsfonts}
\usepackage{amsmath}
\usepackage{amssymb}
\usepackage{float}
\usepackage{url}
\usepackage{hyperref}
\usepackage{subfigure}
\usepackage{here}

\usepackage{tikz}
\usetikzlibrary{arrows, automata, decorations.pathmorphing, fit, shapes}
\tikzstyle{every picture}=[
  >=stealth',
  shorten >=1pt,
  shorten <=1pt,
  node distance=1.5cm,
  auto,
  bend angle=45,
  initial text=,
  every state/.style={
	rounded rectangle, 
	rounded rectangle arc length=90,
  	inner sep=0.5mm,
  	minimum size=5mm
  }
]

\newcommand{\FBA}[1]{\mathcal{F}_{#1}}
\newcommand{\FRL}[1]{\mathcal{L}(#1)}

\newcommand{\LA}[1]{\mathrm{L}_{#1}}
\newcommand{\LsA}[2]{\mathrm{L}_{#1}(#2)}

\newcommand{\CQ}[1]{\Theta({#1})}
\newcommand{\dout}[2]{\delta^{\mathrm{out}}_{#1}(#2)}
\newcommand{\din}[2]{\delta^{\mathrm{in}}_{#1}(#2)}
\newcommand{\OrbAs}[2]{\mathrm{O}_{#1}(#2)}
\newcommand{\AutOrb}[2]{#1_{(#2)}}
\newcommand{\Lin}[2]{\mathrm{L}^{\mathrm{in}}_{#1}(#2)}
\newcommand{\Lout}[2]{\mathrm{L}^{\mathrm{out}}_{#1}(#2)}

\newcommand{\cset}[1]{\mathcal{S}_{#1}}
\newcommand{\AutCut}[2]{#1^{-#2}}
\newcommand{\Lscut}[3]{\mathrm{L}^{-#3}_{#1}(#2)}
\newcommand{\vset}[1]{\mathcal{V}_{#1}}
\newcommand{\AutAdd}[2]{#1^{+#2}}

\begin{document}

\bibliographystyle{splncs_srt}

\title{On the decidability of $k$-Block determinism}

\author{Pascal Caron\inst{1} \and Ludovic Mignot\inst{2} \and Clément Miklarz\inst{1}}

\institute{
    Laboratoire LITIS - EA 4108. Universit\'{e} de Rouen-Normandie, Avenue de l'Universit\'{e} 76801 Saint-\'{E}tienne-du-Rouvray Cedex, France
\and
    Universit\'{e} de Rouen-Normandie, Avenue de l'Universit\'{e} 76801 Saint-\'{E}tienne-du-Rouvray Cedex, France\\
    \email{\{pascal.caron,ludovic.mignot,clement.miklarz1\}@univ-rouen.fr}
}

\maketitle

\begin{abstract}
 	Brüggemann-Klein and Wood define a one-unambi\-guous regular language as a language that can be recognized by a deterministic Glushkov automaton.
 	They give a procedure performed on the minimal DFA, the $\mathrm{BW}$-test, to decide whether a language is one-unambiguous.
	Block determinism is an extension of one-unambi\-guity while considering non-empty words as symbols and prefix-freeness as determinism.
	A block automaton is compact if it does not have two equivalent states (same right language).
	We showed that a language is $k$-block deterministic if it is recognized by some deterministic $k$-block automaton passing the $\mathrm{BW}$-test.
	In this paper, we show that any $k$-block deterministic language is recognized by a compact deterministic $k$-block automaton passing the $\mathrm{BW}$-test.
	We also give a procedure which enumerates, for a given language, the finite set of compact deterministic $k$-block automata. 
	It gives us a decidable procedure to test whether a language is $k$-block deterministic.
\end{abstract}

\section*{Introduction}

    Deterministic or one-unambiguous regular expressions are defined by Brüggemann-Klein and Wood~\cite{BW98} as expressions having a deterministic Glushkov automaton.
    This research has been motivated by the formalization of expressions in Document Type Definition of SGML: one-unambiguity ensures an efficient parsing of these documents.
	The authors characterize languages that can be denoted by such expressions, and show that these languages are strictly included into regular ones.
	Finally, they provide a decidable procedure, the \textrm{BW}-test, to determine whether a given language is one-unambiguous.
	
    Giammaresi {\it et al.}~\cite{GMW01} mention two possible extensions of the notion of one-unambiguity.
    One of them, the block determinism, is linked to block automata defined by Eilenberg~\cite{Eil74} where labels of edges are words.
    In block automata, the notion of determinism is slightly modified since two transitions that start from a same state should not have labels such that one is prefix of the other.
    The authors define a block deterministic language as a language that can be recognized by a block deterministic Glushkov automaton.
    Furthermore, block deterministic languages strictly include one-unambiguous ones.

	Giammarresi{\it et al.} had characterized them using a state elimination procedure on the minimal DFA, but we showed~\cite{CMM16} that one of their lemma is not correct.
	It allowed us to say that state elimination on the minimal DFA is not enough to decide whether a language is $k$-block deterministic.
	Nonetheless, we extended the \textrm{BW}-test on deterministic block automata to state a sufficient condition for a language to be $k$-block deterministic.
	Moreover, we gave a valid proof of the existence of a proper hierarchy in block deterministic languages.
	
	In this paper, we state a necessary and sufficient condition for a language to be $k$-block deterministic.
	Since the minimal deterministic block automata recognizing a language $L$ are not isomorphic, the \textrm{BW}-test cannot be applied on a canonical automaton.
	However, we consider deterministic block automata with no equivalent state, defined as compact automata.
	Then, we show how to compute every compact deterministic block automata recognizing a given regular language.
	Finally, we show that a $k$-block deterministic language is recognized by a compact deterministic $k$-block automaton passing the \textrm{BW}-test.
	Thus, a language is $k$-block deterministic if and only if one of these automata passes the \textrm{BW}-test, and this gives us a decidable procedure to determine whether a language is $k$-block deterministic.

    The paper is organized as follows.
	Section~\ref{se:pre} gathers some preliminary results and notations about languages and automata.
	In Section~\ref{se:bda}, we define block deterministic automata and languages, and characterize them using the \textrm{BW}-test.	
	Section~\ref{sect:compact} is devoted to the computation of the set of compact deterministic $k$-block automata recognizing a given language.
	Finally, we show in Section~\ref{se:compaction} how to compact a deterministic $k$-block automaton while preserving the \textrm{BW}-test.
	This gives us a procedure to decide whether a language is $k$-block deterministic.

\section{Preliminaries}\label{se:pre}

    Let $\Sigma$ be a finite alphabet and $\Sigma^*$ be the set of words over $\Sigma$.
	A \emph{language over $\Sigma$} is a subset of $\Sigma^*$.
    The \emph{length} of a word $w$ is the number $|w|$ of occurrences of symbols of $\Sigma$ appearing in $w$.
    The \emph{empty word} is denoted by $\varepsilon$.
    The set of all prefixes of $w$ is denoted by $\mathrm{Pref}(w)$.
	A language $L$ is \emph{prefix-free} if for every two words $w_1$ and $w_2$ in $L$, $w_1$ is not a prefix of $w_2$.
	Usual operations on sets, like $\cup$, $\cap$, $\setminus$ (set difference) are also defined on languages.
	Let $L$ and $L^{\prime}$ be two languages over $\Sigma$. 
	The \emph{concatenation} $L \cdot L^{\prime}$ is the set $\{w \cdot w^{\prime} \mid w \in L \wedge w^{\prime} \in L^{\prime}\}$ and the \emph{Kleene star} $L^*$ is the set $\bigcup_{k \in \mathbb{N}} L^k$ with $L^0 = \{\varepsilon\}$ and $L^{k+1} = L \cdot L^k$.

	A \emph{regular expression over $\Sigma$} is inductively built from $\emptyset$, $\varepsilon$, and symbols in $\Sigma$ using the binary operators $+$ and $\cdot$, and the unary operator $^*$.
    The \emph{language} $\mathrm{L}(E)$ \emph{denoted} by a regular expression $E$ is inductively defined as follows:
\begin{align*}
	\mathrm{L}(\emptyset) &= \emptyset, & \mathrm{L}(\varepsilon) &= \{\varepsilon\}, & \mathrm{L}(a) &= \{a\},\\
	\mathrm{L}(F + G) &= \mathrm{L}(F) \cup \mathrm{L}(G),\  & \mathrm{L}(F \cdot G) &= \mathrm{L}(F) \cdot \mathrm{L}(G),\  & \mathrm{L}(F^*) &= \mathrm{L}(F)^*,
\end{align*}
with $a \in \Sigma$, and $F$, $G$ some regular expressions over $\Sigma$.	
	The set of \emph{regular languages} is exactly the set of languages that can be denoted by a regular expression.
	A regular expression is \emph{trim} if it is $\emptyset$ or does not contain $\emptyset$.
	We consider only trim regular expressions in the following of this paper.

	An \emph{automaton} $A$ is a 5-tuple $(\Sigma, Q, I, F, \delta)$ defined by $Q$ a finite set of states, $I \subset Q$ the set of initial states, $F \subset Q$ the set of final states, and $\delta \subset Q \times \Sigma \times Q$ the set of transitions.
	The sets defining $A$ are implicitly denoted by $\Sigma_A$, $Q_A$, $I_A$, $F_A$ and $\delta_A$.
	The set $\delta$ is equivalent to a function in $Q \times \Sigma \rightarrow 2^Q$ defined by $(p, a, q) \in \delta \Longleftrightarrow q \in \mathrm{\delta}(p, a)$. 
	This function can be extended to $2^Q \times \Sigma \rightarrow 2^Q$ by $\delta(Q^{\prime}, a) = \bigcup_{q \in Q^{\prime}} \delta(q, a)$.
	The \emph{transitive closure} $\delta^*$ of $\delta$ is the subset $\bigcup_{k \in \mathbb{N}} \delta^k$ of $Q \times \Sigma^* \times Q$ with $\delta^0 = \{(p, \varepsilon, p) \mid p \in Q\}$ and $\delta^{k+1} = \{(p, au, q) \mid (p, a, r) \in \delta \wedge (r, u, q) \in \delta^k\}$.
	The \emph{set of final labels} of $A$ is the set of labels of the transitions going out of the final states of $A$, denoted by $\FBA{A} = \{a \mid \delta(F, a) \neq \emptyset\}$.
	The \emph{right language} $L(q)$ of a state $q$ in $Q$ is the set $\{w \in \Sigma^* \mid \delta^*(q, w) \cap F \neq \emptyset\}$.
	Two states are equivalent if they have the same right language.
	The \emph{family of right languages} of $A$ is denoted by $\FRL{A} = \{\LsA{}{q} \mid q \in Q\}$.
	The \emph{language recognized by $A$} is the set $L_A = \{w \in \Sigma^* \mid \delta^*(I,w) \cap F \neq \emptyset\}$.
	Two automata are \emph{equivalent} if they recognize the same language.	
	Kleene's Theorem~\cite{Kle56} asserts that the set of regular languages is the same as the set of languages recognized by finite automata.
	Several algorithms to compute an automaton from a regular expression have been given, such as the Glushkov one~\cite{Glu61}.		
	Let $p$ and $q$ be two states in $Q$.
	Then $q$ is \emph{reachable from} $p$ if there exists a path of transitions $(p, w, q)$ in $\delta^*$.
	The state $p$ is \emph{accessible} (resp. \emph{co-accessible}) in $A$ if there exists a state $i$ in $I$ (resp. $f$ in $F$) such that $p$ is reachable from $i$ (resp. $f$ is reachable from $p$).
    The automaton $A$ is \emph{trim} if all its states are both accessible and co-accessible.
    The automaton $A$ is \emph{deterministic} (and called a DFA) if $|I| = 1$ and for any two distinct transitions $(p, a, q_1)$ and $(p, b, q_2)$ in $\delta_A$, $a$ is different from $b$.
    A DFA is \emph{minimal} if there is no equivalent DFA with fewer states.       
    Two equivalent minimal DFA are isomorphic and consequently, there exists a canonical DFA recognizing a regular language $L$ called the minimal DFA of $L$.
	The minimal DFA of a language can be computed from a trim DFA by merging equivalent states, and thus its states are pairwise not equivalent.
    The automaton $A$ is \emph{residual} if its family of right languages is included in the one of its minimal DFA.	 

  	A non-empty subset of states of an automaton is an \emph{orbit} if it is a strongly connected component.   
    The \emph{orbit of a state $q$}, denoted by $\OrbAs{}{q}$, is the strongly connected component to which $q$ belongs.  
    An orbit is \emph{trivial} if it consists of only one state with no self-loop.
    A \emph{non re-entering component} $R$ is a union of orbits such that for any state $r$ in $R$ and $q$ not in $R$, if $r$ can reach $q$ then $q$ cannot reach $r$.
    
	The set of out-transitions (resp. in-transitions) of $R$ is denoted by $\dout{}{R} = \{(r, a, p) \in \delta \mid r \in R \wedge p \notin R\}$ (resp. $\din{}{R} = \{(p, a, r) \in \delta \mid p \notin R \wedge r \in R\}$).
	The internal transitions of $R$ are the transitions linking any two states of $R$.
	A state of $R$ is a \emph{gate} if it is either final or the origin of an out-transition.
  
	The notions of reachability, accessibility and co-accessibility is extended to orbits in relation to their states.

\section{Block determinisitic languages}\label{se:bda}

    A \emph{block} is a non-empty word over $\Sigma$.
	The \emph{set of blocks of $\Sigma$} is denoted by $\Gamma_{\Sigma}$.	
	The notions of regular expression and automaton over an alphabet $\Sigma$ can be extended to block regular expressions and block automaton~\cite{Eil74,GM99} replacing symbols of $\Sigma$ by blocks of $\Gamma_{\Sigma}$.
    Notice that the transitive closure $\delta^*$ is still a subset of $(Q \times \Sigma^*\times Q)$.
    Since $\Gamma_{\Sigma}$ is a subset of $\Sigma^*$, the distinction between $\delta$ and $\delta^*$ is needed to differentiate a simple transition from a path of transitions.
	
    To distinguish blocks as syntactic components in a block regular expression, they are written between square brackets.
	Those are omitted for one letter blocks.
	Moreover, blocks can be treated as single symbols, as we do when we refer to the elements of an alphabet.
    This allows us to construct the Glushkov (block) automaton of a block regular expression.

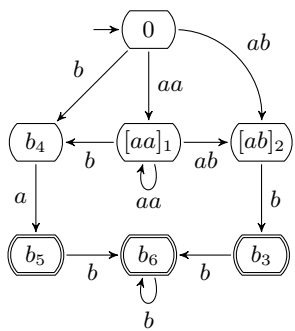
\begin{figure}
	\centering	
	\begin{tikzpicture}
		\node[state, initial] (i) {$0$};
	    \node[state, below of=i] (1) {$[aa]_1$};
	    \node[state, left of=1] (4) {$b_4$};
	    \node[state, right of=1] (2) {$[ab]_2$};
	    \node[state, accepting, below of=2] (3) {$b_3$};
	    \node[state, accepting, below of=4] (5) {$b_5$};
  	  	\node[state, accepting, below of=1] (6) {$b_6$};
	    \path[->]
    	        (i) edge node {$aa$} (1)
            	(i) edge [swap] node {$b$} (4)
       		(i) edge [bend left=45] node {$ab$} (2)
            (1) edge [loop below] node {$aa$} ()
    	        (1) edge node {$b$} (4)
   		    (1) edge [swap] node {$ab$} (2)
   		    (2) edge node {$b$} (3)
  		    (3) edge node {$b$} (6)
   		    (4) edge [swap] node {$a$} (5)
   		    (5) edge [swap] node {$b$} (6)
   		    (6) edge [loop below] node {$b$} ()
		;
	\end{tikzpicture}
	\caption{The Glushkov automaton of $E = [aa]^*([ab]b + ba)b^*$}
	\label{fg:GlushkovBloc}
\end{figure}

    We denote by $\Gamma_E$ (resp. $\Gamma_B)$ the set of blocks appearing in $E$ (resp. labelling transitions of $B$).
	Let $E$ be a block regular expression and $A$ be a block automaton such that $\Gamma_E = \Gamma_A = \Gamma$, then $E$ and $A$ are \emph{$k$-block} for some integer $k$ if the length of any block $b$ in $\Gamma$ is smaller than or equal to $k$.
	Since $\Sigma \subset \Gamma_{\Sigma}$, regular expressions (resp. automata) are $1$-block regular expressions (resp. $1$-block automata).
	Thus, we refer to regular expressions and automata as being block regular expressions and block automata.

	The notion of determinism is extended to block automata as follows:
\begin{definition}	
	A block automaton $A$ is \emph{deterministic} if $|I| = 1$ and for any two distinct transitions $(p, b_1, q_1)$ and $(p, b_2, q_2)$ in $\delta$, $b_1$ is not a prefix of $b_2$.
\end{definition}

    A regular expression E is \emph{one-unambi\-guous}~\cite{BW98} if its Glushkov automaton is a DFA, and a language is \emph{one-unambiguous} if it can be denoted by a one-unambiguous regular expression.
	Block deterministic languages~\cite{GMW01} are an extension of one-unambi\-guous languages: a block regular expression $E$ is \emph{deterministic} if its Glush\-kov automaton is deterministic, and a language is \emph{$k$-block deterministic} if it can be denoted by a deterministic $k$-block regular expression.
	As an example, the Glushkov automaton in Figure~\ref{fg:GlushkovBloc} is deterministic and thus, the language $\mathrm{L}([aa]^*([ab]b + ba)b^*)$ is $2$-block deterministic.
	The family of one-unambiguous languages is the same as the family of $1$-block deterministic languages, and we showed~\cite{CMM16} that there is a proper infinite hierarchy in $k$-block deterministic languages.
	
	The one-unambiguity of a language is structurally decidable over its minimal DFA.
	The decision procedure is related to the orbits of its underlying graph and to their links with the remaining parts.
		
	A non re-entering component $R$ of an automaton $A$ is \emph{transverse} if for any two gates $p$ and $q$ of $R$, for any symbol $a$, for any state $r$ of $A$, $p$ is final if $q$ is, and $(p,a,r)$ is an out-transition if $(q,a,r)$ is.
	The automaton $A$ has the \emph{orbit property} if all its orbits are transverse.

	The \emph{orbit automaton $\AutOrb{A}{q}$ of the state $q$} is the automaton obtained by restricting the states of $A$ to $\OrbAs{}{q}$ and the transitions of $A$ to the internal transitions of the orbit $\OrbAs{}{q}$ while considering $q$ as the initial state and  the gates of $\OrbAs{}{q}$ as the final states.
	For any state $q$ of $A$, the language $\LA{\AutOrb{A}{q}}$ is an \emph{orbit language of $A$}.
	
	A label $a$ of $\Gamma_A$ is \emph{$A$-consistent} if there exists a state $q_a$ of $A$ such that every final state of $A$ has an outgoing transition labelled by $a$ to $q_a$, and those outgoing transitions are called \emph{synchronizing transitions}.
	The \emph{set of consistency of $A$} is denoted by $\cset{A} = \{(a, q_a) \mid \forall f \in F, (f, a, q_a) \in \delta\}$.
	This definition is trivially extended to an orbit as the set of consistency of any of its orbital automaton.
	Let $s = (a, q_a)$ be an element of $\cset{A}$.
	The \emph{$s$-cut} $\AutCut{A}{s}$ of $A$ is constructed from $A$ by removing for each final state $f$ of $A$, every transition $(f, a, q_a)$ of $A$.
	It is naturally extended to a subset of $\cset{A}$.

	Brüggemann-Klein and Wood define an inductive algorithm and give a characterization of the one-unambiguous languages.

%

    However, we use a slightly different but equivalent test, which we name the \emph{\textrm{BW}-test}. 

\begin{theorem}[\cite{BW98}]\label{th:BWtest}
    Let $M$ be a minimal DFA.
    Then, $\LA{M}$ is one-unambiguous if and only if $M$ has the orbit property and all orbit languages of $M$ are one-unambiguous.
    If $\LA{M}$ is one-unambiguous, then a one-unambiguous regular expression denoting $\LA{M}$ can be constructed from one-unambiguous regular expressions for the orbit languages.
\end{theorem}

	Thus, let $A$ be an automaton, the \textrm{BW}-test is performed as follows:
\begin{enumerate}
	\item if $A$ does not have the orbit property, the test halts and fails
	\item for each non trivial orbit of $A$:
		\begin{enumerate}
		    \item choose one of its orbital automaton $B$
			\item if $\cset{B}$ is empty, the test halts and fails
			\item choose $(b, s)$ in $\cset{B}$
			\item recursively test $\AutCut{(\AutOrb{B}{s})}{(b, s)}$
		\end{enumerate}
	\item the test succeeds
\end{enumerate}

    Notice that this test always terminates.
    Indeed, if it never fails, by selecting orbital automata and removing their transitions with a consistent label going out of their final states, we necessarily get acyclic automata at some point (which contain only trivial orbits).
    Moreover, removing useless states preserves the \textrm{BW}-test.
    
    Let $A$ be a block automaton.
	The automaton $B$ is \emph{an alphabetic image of} $A$ if $Q_B=Q_A$, $I_B=I_A$, $F_B=F_A$ and there exists an injection $\phi$ from $\Gamma_A$ to $\Gamma_B$ such that $\delta_B = \{(p, \phi(a), q) \mid (p, a, q) \in \delta_A\}$.
	
\begin{theorem}[\cite{CMM16}]\label{th:KBD}
	A language is $k$-block deterministic if and only if it is recognized by a deterministic $k$-block automaton $B$ such that $B$ is the alphabetic image of a DFA passing the \textrm{BW}-test.
\end{theorem}

	As the \textrm{BW}-test is purely structural, it can be directly applied to block automata while considering blocks as symbols.
	Thus, this theorem allows us to compute a deterministic $k$-block regular expression denoting the language recognized by a deterministic $k$-block automaton passing the \textrm{BW}-test.

	Notice that, if the test fails, we cannot decide whether the language is $k$-block deterministic or not.
	Moreover, Giammarresi and Montalbano show~\cite{GM99} that, in general, there is not a unique minimal deterministic block automaton recognizing a language.
    Thus, unlike DFA, minimality in deterministic block automata is not sufficient. 
    In this paper, we focus on a more general concept of block automata which includes the case of minimal DFA.
    
\begin{definition}
    A block automaton $A$ is \emph{compact} if it has no two equivalent states.
\end{definition}
    
    Similarly, an orbit of a block automaton is compact if it does not contain two equivalent states.
        
    First, we prove that, given a language $L$ and an integer $k$, there is a finite number of trim compact deterministic $k$-block automata recognizing $L$ which can be computed from its minimal DFA.
    Then, we describe a procedure to compute a compact block automaton from a deterministic $k$-block automaton while preserving the \textrm{BW}-test.

\section{Computation of compact deterministic block automata}\label{sect:compact}

    In this section, we first give some properties related to deterministic block automata which are necessary to define the $k$-transition automaton. 
    Then we prove that given a language $L$, any compact deterministic block automaton recognizing $L$ is a sub-automaton of the $k$-transition automaton which can be computed from its minimal DFA.

\begin{lemma}\label{lm:state_kbd}
    Let $A$ be a trim deterministic block automaton and $M$ be the minimal DFA of $L_A$.
    For any state $p$ of $A$, there exists a state $q$ of $M$ such that $\LsA{}{p} = \LsA{}{q}$.
\end{lemma}
\begin{proof}
    Let $p$ be a state of $A$.
    There exists two words $w_p$ and $w_s$ in $\Sigma_A^*$, such that $p$ belongs to $\delta_A^*(I_A, w_p)$ and $\delta_A^*(p, w_s) \cap F_A \neq \emptyset$.
    Thus, $w_pw_s$ belongs to $\LA{A}$.
    If $\delta_M^*(I_M, w_p) = \emptyset$, then $w_pw_s$ does not belong to $\LA{M}$ which is contradictory.
    Therefore, there exists a state $q$ in $\delta_M^*(I_M, w_p)$.
    Let $w_1$ be a word in $\LsA{A}{p}$, then $w_pw_1$ belongs to $\LA{A} = \LA{M}$ and $w_1$ belongs to $\LsA{M}{q}$.
    Let $w_2$ be a word not in $\LsA{A}{p}$.
    Let us suppose that $w_pw_2$ belongs to $\LA{A}$.
    Since there is only one initial state, there would exist a state $s$ in $Q_A$ with two out-going transitions which are prefix from each other, contradicting the determinism.
    Thus, $w_pw_2$ does not belong to $\LA{A} = \LA{M}$ and $w_2$ does not belong to $\LsA{M}{q}$.
    Therefore, $p$ and $q$ are equivalent.
\end{proof}

    Thus, we can define the function $\Phi$ which associates to each state in $A$ its equivalent state in $M$. 
    It is naturally extended to set of states.
    
    Let us notice that $\Phi$ is not necessarily surjective.
    Moreover, unlike two equivalent DFA, two equivalent deterministic block automata may have different families of right languages.
    Thus, two block automata are \emph{$\mathcal{L}$-equivalent} if they are equivalent and have the same family of right languages.

    Since for any state of $A$, at most one state is reached for any word, it follows that:
 
\begin{corollary}\label{coro:state_kbd}
    Let $M$ be the minimal DFA of a trim deterministic block automaton $A$.
    If $(p, w, q)$ belongs to $\delta_A^*$, then $(\Phi(p), w, \Phi(q))$ belongs to $\delta_M^*$.
\end{corollary}

    Thus, if two states $p$ and $q$ of $A$ belong to the same orbit, then $\Phi(p)$ and $\Phi(q)$ also belong to the same orbit in $M$.
    So, we can define the function $\Omega$ which associates to any orbit $O$ of $A$ an orbit $K$ of $M$ such that $\Phi(O)$ is a subset of $K$, and the function $\CQ{O} = \Omega(O) \setminus \Phi(O)$ which represents the set of states of $\Omega(O)$ that have no equivalent in $O$.
    Then, an orbit $O$ of $A$ is \emph{maximal} if for any distinct orbit $O'$ such that $O'$ is reached from $O$, we have $\Omega(O) \neq \Omega(O')$.
   
    Considering only final states, the function $\Phi$ is surjective:

\begin{lemma}\label{lm:fstate_reachability}
    Let $M$ be the minimal DFA of a trim deterministic block automaton $A$.
    Then, we have $\Phi(F_A) = F_M$.
\end{lemma}
\begin{proof}
    Let $p$ be a state of $A$ and $f_M$ be a final state of $M$ such that $(\Phi(p), w, f_M)$ belongs to $\delta_M^*$.
    Since $p$ and $\Phi(p)$ are equivalent, there exists a final state $f_A$ of $A$ such that $(p, w, f_A)$ belongs to $\delta_A^*$.
    Following Corollary~\ref{coro:state_kbd}, $(\Phi(p), w, \Phi(f_A))$ belongs to $\delta_M^*$, which means that $\Phi(f_A) = f_M$.
\end{proof}

    The function $\Phi$ can be trivially extended from trim block automata to non-necessarily trim residual ones and consequently Lemma~\ref{lm:state_kbd} and Lemma~\ref{lm:fstate_reachability} still hold for residual block automata.   
    
    Then, let us show that the set of trim compact deterministic block automata is finite and computable.
    We define a super-automaton, including all these candidates, obtained as a finite closure of the transition function with respect to the words of length at most $k$.
    Following Lemma~\ref{lm:fstate_reachability}, for the language and the block determinism to be preserved, a final state should not be avoided.
    
\begin{definition}
	Let $L$ be a language and $M$ be its minimal DFA.
	The $k$-\emph{transition automaton} $T$ of $L$ is defined by
	$$\begin{array}{l}
	\Sigma_T=\Sigma_M,\ Q_T=Q_M,\ I_T=I_M,\ F_T=F_M, \\
	 \delta_T = \{(p, w, r) \in \delta^*_M \mid 1 \leq |w| \leq k \wedge \forall u \in \mathrm{Pref}(w) \setminus \{\varepsilon, w\}, u \notin \LsA{M}{p}\}.
	 \end{array}$$
\end{definition}    
    
    As an example, let us consider the minimal DFA $M$ in Figure~\ref{figMinDFA}.
    The $2$-transition automaton $T$ of $\LA{M}$ is presented Figure~\ref{ExempleATU}.
    Notice that any transition in $T$ is obtained by extending some transitions of $M$, except if a final state is reached. 
    Thus, the transitions $(1, aa, i)$ and $(1, ab, 1)$ do not exist since they would go through the final state $2$.
    
\begin{figure}[H]
    \begin{minipage}[b]{.43\linewidth}
        \centering
        \begin{tikzpicture}[node distance=2cm]
			\node[state, initial] (i) {$i$};
			\node[state, above left of=i] (1) {$1$};
			\node[state, accepting, above right of=i] (2) {$2$};
			\path[->]
				(i) edge node {$a$} (1)
				(1) edge [loop left] node {$b$} ()
				(1) edge [bend left=25] node {$a$} (2)
				(2) edge [bend left=25] node {$b$} (1)
				(2) edge node {$a$} (i)
			;
		\end{tikzpicture}
        \caption{The minimal DFA $M$}
        \label{figMinDFA}
    \end{minipage}
    \hfill
    \begin{minipage}[b]{.56\linewidth}
        \centering
        \begin{tikzpicture}[node distance=2.5cm]
			\node[state, initial] (i) {$i$};
			\node[state, above left of=i] (1) {$1$};
			\node[state, accepting, above right of=i] (2) {$2$};
			\path[->]
				(i) edge node {$a, ab$} (1)
				(i) edge [bend left=15] node {$aa$} (2)
				(1) edge [loop left] node {$b, bb$} ()
				(1) edge [bend left=15] node {$a, ba$} (2)
				(2) edge [bend left=15,swap] node {$aa, bb, b$} (1)
				(2) edge [loop right] node {$ba$} ()
				(2) edge [bend left=15] node {$a$} (i)
				;
        \end{tikzpicture}
        \caption{The $2$-transition automaton $T$ of $\LA{M}$}
        \label{ExempleATU}
    \end{minipage}
\end{figure}

    We show that any compact deterministic $k$-block automaton recognizing a language $L$ can be computed from the $k$-transition automaton of $L$:

\begin{proposition}\label{prop:tcKBD_KTA}
    Let $B$ be a trim deterministic $k$-block automaton of a language $L$ and $T$ be the $k$-transition automaton of $L$.
    Then $B$ is compact if and only if it is isomorphic to a sub-automaton of $T$.
\end{proposition}
\begin{proof}
    Let $M$ be the minimal DFA of $L$. 

    Let us suppose that $B$ is a sub-automaton of $T$.
    Let $q$ be a state of $B$ and $w$ be a word in $\Sigma^*$ such that $q$ belongs to $\delta_B^*(I_B, w)$.
    Since $B$ is a sub-automaton of $T$, $q$ also belongs to $\delta_T^*(I_T, w)$.
    By definition of the $k$-transition automaton, $\delta_M \subset \delta_T \subset \delta_M^*$, which means that $\delta_T^* = \delta_M^*$.
    Thus, $q$ belongs to $\delta_M^*(I_M, w)$, and since $B$ is deterministic and recognizes $L$, $\Phi_B(q) = q$.
    Since $M$ is compact, so is $B$.
    
    Now, let us suppose that $B$ is compact.
    Since $B$ is deterministic and recognizes $L$, we have $\Phi(I_B) = I_M = I_T$ and following Lemma~\ref{lm:state_kbd} and Lemma~\ref{lm:fstate_reachability}, we have $\Phi(Q_B) \subset Q_M = Q_T$ and $\Phi(F_B) = F_M = F_T$.
    As $B$ is trim and compact, for any state $s$ of $M$ (and thus of $T$), there is at most one state $r$ of $B$ such that $\Phi(r)= s$ and thus $|Q_B| \leq |Q_T|$.
    Moreover, following Corollary~\ref{coro:state_kbd}, for any transition $(p, w, q)$ in $\delta_B$, $(\Phi(p), w, \Phi(q))$ belongs to $\delta_M^*$.
    Since $B$ and $T$ are both $k$-block, by construction, $(\Phi(p), w, \Phi(q))$ belongs to $\delta_T$.
    Thus, $B$ is isomorphic to a sub-automaton of $T$.
\end{proof}

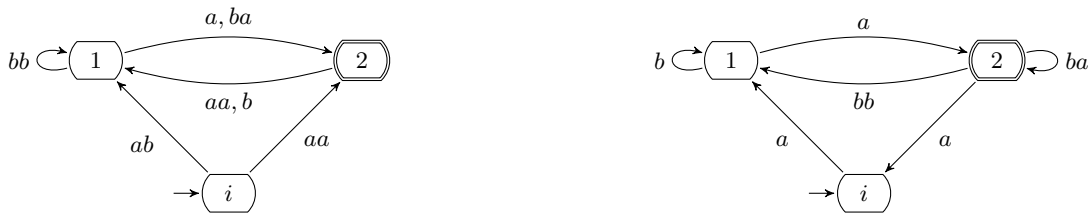
\begin{figure}
    \begin{minipage}[b]{.44\linewidth}
        \centering
        \begin{tikzpicture}[node distance=2.5cm]
			\node[state, initial] (i) {$i$};
			\node[state, above left of=i] (1) {$1$};
			\node[state, accepting, above right of=i] (2) {$2$};
			\path[->]
				(i) edge node {$ab$} (1)
				(i) edge [swap] node {$aa$} (2)
				(1) edge [loop left] node {$bb$} ()
				(1) edge [bend left=15] node {$a, ba$} (2)
				(2) edge [swap, bend left=15,swap] node {$aa, b$} (1)
				;
        \end{tikzpicture}
    \end{minipage}
    \hfill
    \begin{minipage}[b]{.54\linewidth}
        \centering
        \begin{tikzpicture}[node distance=2.5cm]
			\node[state, initial] (i) {$i$};
			\node[state, above left of=i] (1) {$1$};
			\node[state, accepting, above right of=i] (2) {$2$};
			\path[->]
				(i) edge node {$a$} (1)
				(1) edge [loop left] node {$b$} ()
				(1) edge [bend left=15] node {$a$} (2)
				(2) edge [swap, bend left=15,swap] node {$bb$} (1)
				(2) edge [loop right] node {$ba$} ()
				(2) edge node {$a$} (i)
				;
        \end{tikzpicture}
    \end{minipage}
    \caption{Two sub-automata of $T$ recognizing $L(T)$}
    \label{fg:ExPrefMax}
\end{figure}

    Notice that the set of sub-automata of a $k$-transition automaton of a language is finite and computable by brute force algorithms.
    Thus,
    
\begin{theorem}\label{thm:KBD_FIN}
    The set of trim compact deterministic $k$-block automata recognizing a language $L$ is finite and computable.
\end{theorem}

%

\section{Compaction of deterministic block automata passing the \textrm{BW}-test}\label{se:compaction}

    Giammarresi and Montalbano~\cite{GM99} give a procedure to get a compact block automaton from a deterministic one, by merging equivalent states and selecting a subset of transitions.
    Unfortunately, this procedure does not necessarily preserve the \textrm{BW}-test.

    In this section, we fix this problem.
    The idea underlying this procedure applied on a block automaton follows the definition of the \textrm{BW}-test, and can be decomposed in two main steps:
\begin{itemize}
    \item[$\bullet$] Transforming each orbit into a compact one by removing the synchronizing transitions of its orbital automaton, computing an equivalent compact automaton and reintroducing the removed transitions.
    \item[$\bullet$] Selecting a minimal subset of compact orbits representing the set of orbits of its minimal DFA.
\end{itemize}

    In this section, we describe and prove the properties of the different steps of our procedure.
    A running example is given in Figure~\ref{fig:ExAlgoA} to enlighten the purpose.


\begin{figure}[H]
    \centering
    \begin{tikzpicture}
		\node[state, initial] (i) {$i_A$};
		\node[state, above right of=i] (1) {$1$};
		\node[state, above right of=1] (3) {$3$};
		\node[state, below right of=1] (3') {$3'$};
		\node[state, below right of=3] (4) {$4$};
		\node[state, below left of=3'] (5) {$5$};
		\node[state, below of=5] (2) {$2$};
		\node[state, right of=2] (3'') {$3''$};
		\node[state, accepting, right of=4, xshift=0.5cm] (f1) {$f_1$};
		\node[state, accepting, right of=f1] (f2) {$f_2$};		
		\path[->]
			(i) edge node {$a$} (1)
			(i) edge [swap, bend right=45] node {$ba$} (2)
			(1) edge node {$aa$} (3)
			(1) edge [swap] node {$ab$} (3')
			(3) edge node {$a$} (4)
			(3') edge [swap] node {$a$} (4)			
			(4) edge [swap] node {$aab, aaa$} (1)
			(4) edge node {$ab, aac$} (f1)
			(2) [swap] edge node {$a, b$} (3'')
			(3'') edge node {$aa$} (5)
			(5) edge node {$aaa$} (2)
			(5) edge [swap] node {$ab$} (1)
			(5) edge [bend right=25] node {$b, ac$} (f1)
			(f1) edge [swap] node {$a$} (f2)
			(f2) edge [loop above] node {$a$} ()
		;
	\end{tikzpicture}
	\caption{The automaton $A$}
    \label{fig:ExAlgoA}
\end{figure}
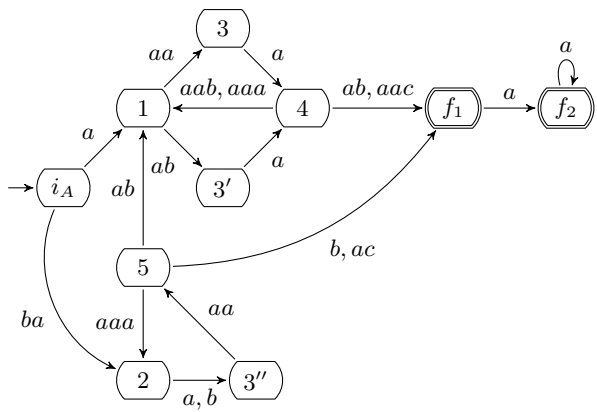

\subsection{Extracting and substituting an orbit}\label{ss:ExtSubsOrb}

    Extracting an orbit $O$ of an automaton $A$ consists in computing one of its orbital automaton.
    Obviously, this operation preserves the determinism and the maximal size of the blocks.
    
    Moreover, following Theorem~\ref{th:BWtest} and the description of the \textrm{BW}-test, it holds that:
    
\begin{corollary}\label{coro:BWt_orbAut}
    Let $q$ be a state of a block automaton $A$.
    If $A$ passes the \textrm{BW}-test, then so does $\AutOrb{A}{q}$.
\end{corollary}
    
\begin{example}
    In the automaton $A$ of Figure~\ref{fig:ExAlgoA}, the orbit $\{1, 3, 3', 4\}$ is not compact.
    Thus, we proceed to its extraction to get the automaton $B = \AutOrb{A}{1}$ of Figure~\ref{fig:ExAlgoB}.
\end{example}

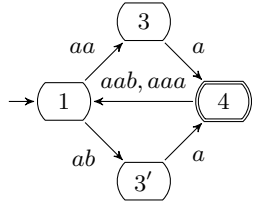
\begin{figure}[H]
    \centering
    \begin{tikzpicture}
		\node[state, initial] (1) {$1$};
		\node[state, above right of=1] (3) {$3$};
		\node[state, below right of=1] (3') {$3'$};
		\node[state, accepting, below right of=3] (4) {$4$};
		\path[->]
            (1) edge node {$aa$} (3)
            (1) edge [swap] node {$ab$} (3')
            (3) edge node {$a$} (4)
            (3') edge [swap] node {$a$} (4)			
            (4) edge [swap] node {$aab, aaa$} (1)
		;
	\end{tikzpicture}
	\caption{The automaton $B$}
    \label{fig:ExAlgoB}
\end{figure}
    
	Let $R$ be a non re-entering component.
	The \emph{internal language} of a state $p$ in $R$, denoted by $\Lin{R}{p}$, is the set of words which send $p$ to the gates of $R$.
	The \emph{external language} of $p$, denoted by $\Lout{R}{p}$, is the set of words which send $p$ to a final state without using any internal transition in $R$.
	Thus, $\Lout{R}{p} = \emptyset$ if and only if $p$ is not a gate.
    If $R$ is transverse, then for any two gates $g_1$ and $g_2$ of $R$, we have $\Lout{R}{g_1} = \Lout{R}{g_2}$.
	In this case, we define $\Lout{}{R}=\Lout{R}{g_1}$ as \emph{the external language of $R$}.

\begin{lemma}\label{lm:lso}
	Let $R$ be a transverse non re-entering component of a block automaton $A$.
	Then, for every state $p$ in $R$, we have $\LsA{}{p} = \Lin{R}{p} \cdot \Lout{}{R}$.
\end{lemma}
\begin{proof}
    Let $w$ be a word.
    Then $w$ belongs to $\LsA{}{p}$ if and only if there exists a final state $f$ such that $(p, w, f)$ belongs to $\delta^*$.
	If $f$ belongs to $R$, then $f$ is a gate, which means that $w$ belongs to $\Lin{R}{p}$ and $\varepsilon$ belongs to $\Lout{}{R}$.
	Otherwise, there exists a gate $g$ of $R$ such that $(p, u, g)$ and $(g, v, f)$ belong to $\delta^*$ with $w = uv$, which means that $u$ belongs to $\Lin{R}{p}$ and $v$ belongs to $\Lout{R}{g}$, that is to say $\Lout{}{R}$.
\end{proof}

\begin{proposition}\label{prop:lorb}
	Let $R$ be a transverse non re-entering component of a deterministic block automaton $A$.
	Then, for any two states $p$ and $q$ in $R$, we have $\LsA{}{p} = \LsA{}{q}$ if and only if $\Lin{R}{p} = \Lin{R}{q}$.
\end{proposition}
\begin{proof}
	Let us suppose that $\Lin{R}{p} = \Lin{R}{q}$.
	Since $R$ is transverse, from Lemma \ref{lm:lso}, $\LsA{}{p} = \Lin{R}{p} \cdot \Lout{}{R} = \Lin{R}{q} \cdot \Lout{}{R} = \LsA{}{q}$.
	
	Let us suppose that $\LsA{}{p} = \LsA{}{q}$ and $\Lin{R}{p} \neq \Lin{R}{q}$.
    Then, there exists a shortest word $w$ in $\Lin{R}{p}$ but not in $\Lin{R}{q}$. 
    We have two cases to consider:
\begin{enumerate}
    \item Each proper prefix of $w$ is neither in $\Lin{R}{p}$ nor in $\Lin{R}{q}$. 
    Let $w_o$ be a shortest word of $\Lout{}{R}$.
    Then there exist two words $u_o, v_o$ such that $w_o=u_ov_o$, $wu_o\in \Lin{R}{q}$ and $v_o \in \Lout{}{R}$.
    However, $w_o$ is a shortest word of $\Lout{}{R}$, so $w_o=v_o$, $u_o=\varepsilon$ and $w \in \Lin{R}{q}$. 
    Contradiction.
    \item Let $u$ be the longest proper prefix of $w$ such that $u \in \Lin{R}{p} \cap \Lin{R}{q}$. 
    Let $v$ be the word such that $w = uv$ and let $u_o$ be a shortest word of $\Lout{}{R}$. 
    As $wu_o\in \LsA{}{q}$ and as $u_o$ is a shortest word, we have $\delta^*(q,w) \cap R = \emptyset$ and $vu_o \in \Lout{}{R}$.
    Let $g \in \delta^*(p,u)$. 
    As $u \in \Lin{R}{p}$, $g$ is a gate and as $w=uv \in \Lin{R}{p}$, $v \in \Lin{R}{g}$. 
    But $v$ is a prefix of a word of $\Lout{}{R}$, so the automaton is not deterministic. 
    Contradiction. 
\end{enumerate}
\end{proof}


    Now, let us study the reverse operation of extraction, the substitution of an orbit.
    
    Let $A$ be a block automaton having the orbit property, $q$ be a state of $A$, and $B$ be a block automaton $\mathcal{L}$-equivalent to $\AutOrb{A}{q}$.
    Since $B$ and $\AutOrb{A}{q}$ are $\mathcal{L}$-equivalent, there exists at least one function $h$ which associates each state $p$ of $\AutOrb{A}{q}$ 
    with a state of $B$ equivalent to $p$.

    A \emph{substitution of $B$ for $\OrbAs{}{q}$} constructs a block automaton $C$ defined as follows:  
\begin{itemize}
    \item $\Sigma_C = \Sigma_A \cup \Sigma_B$
	\item $Q_C = (Q_A \setminus \OrbAs{}{q}) \cup Q_B$
	\item $I_C = \{h(i_A)\}$ if $i_A \in \OrbAs{}{q}$, $\{i_A\}$ otherwise
	\item $F_C = (F_A \setminus \OrbAs{}{q}) \cup F_B$ if $F_A \cap \OrbAs{}{q} \neq \emptyset$, $F_A$ otherwise
	\item $\delta_C =$
        \begin{tabular}[t]{l@{\ }l}
    	  $(\delta_A \setminus (Q_A \times \Gamma_A \times \OrbAs{}{q}) \setminus (\OrbAs{}{q} \times \Gamma_A \times Q_A))$\\
	      $\cup\ \delta_B$\\
	      $\cup\ \{(p, b, h(o)) \mid (p, b, o) \in \din{}{\OrbAs{}{q}}\}$\\
	      $\cup\ \{(f_B, b, p) \mid f_B \in F_B \wedge \exists (o, b, p) \in \dout{}{\OrbAs{}{q}}\}$
        \end{tabular}
\end{itemize}    

\begin{example}
    The automaton $B'$ of Figure~\ref{fig:ExAlgoB'} is $\mathcal{L}$-equivalent to the automaton $B = \AutOrb{A}{1}$ of Figure~\ref{fig:ExAlgoB}.
    Then, the automaton $A'$ of Figure~\ref{fig:ExAlgoA'} is obtained from the automaton $A$ of Figure~\ref{fig:ExAlgoA} by substituting $B'$ for $\OrbAs{}{1}$.
\end{example}

\begin{figure}[H]
    \centering
    \begin{tikzpicture}
		\node[state, initial] (1) {$1$};
		\node[state, right of=1] (3) {$3$};
		\node[state, accepting, right of=3] (4) {$4$};
		\path[->]
            (1) edge [swap] node {$aa, ab$} (3)
            (3) edge [swap] node {$a$} (4)
            (4) edge [swap, bend right=25] node {$aab, aaa$} (1)
		;
	\end{tikzpicture}
	\caption{The automaton $B'$}
    \label{fig:ExAlgoB'}
\end{figure}
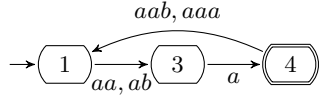

\begin{figure}[H]
    \centering
    \begin{tikzpicture}
		\node[state, initial] (i) {$i_A$};
		\node[state, above right of=i] (1) {$1$};
		\node[state, right of=1] (3) {$3$};
		\node[state, right of=3] (4) {$4$};
		\node[state, below right of=i] (5) {$5$};
		\node[state, below of=5] (2) {$2$};
		\node[state, right of=2] (3'') {$3''$};
		\node[state, accepting, below right of=4] (f1) {$f_1$};
		\node[state, accepting, right of=f1] (f2) {$f_2$};		
		\path[->]
			(i) edge node {$a$} (1)
			(i) edge [swap, bend right=45] node {$ba$} (2)
			(1) edge [swap] node {$aa, ab$} (3)
            (3) edge [swap] node {$a$} (4)
            (4) edge [swap, bend right=25] node {$aab, aaa$} (1)
			(4) edge node {$ab, aac$} (f1)
			(2) [swap] edge node {$a, b$} (3'')
			(3'') edge node {$aa$} (5)
			(5) edge node {$aaa$} (2)
			(5) edge [swap] node {$ab$} (1)
			(5) edge node {$b, ac$} (f1)
			(f1) edge [swap] node {$a$} (f2)
			(f2) edge [loop above] node {$a$} ()
		;
	\end{tikzpicture}
	\caption{The automaton $A'$}
    \label{fig:ExAlgoA'}
\end{figure}
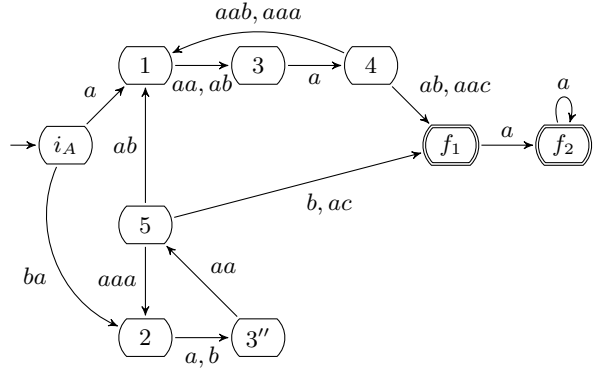
    
    Let us study the properties that are preserved by the substitution of an orbit:
    
\begin{proposition}\label{prop:prop_subst}
    Let $q$ be a state of a deterministic $k$-block automaton $A$ which passes the \textrm{BW}-test.
    Let $B$ be a deterministic $k$-block automaton $\mathcal{L}$-equivalent to $\AutOrb{A}{q}$, such that $\FBA{B}$ is a subset of $\FBA{\AutOrb{A}{q}}$.
    Let $C$ be a block automaton constructed by substituting $B$ for $\OrbAs{}{q}$, then:
    \begin{enumerate}
        \item\label{subst_i1} $C$ is $k$-block
        \item\label{subst_i2} $\FBA{C}$ is a subset of $\FBA{A}$
        \item\label{subst_i3} $C$ is deterministic
        \item\label{subst_i4} $C$ is $\mathcal{L}$-equivalent to $A$
        \item\label{subst_i5} if $B$ passes the \textrm{BW}-test, then so does $C$
        \item\label{subst_i6} if $B$ is compact, then every orbit of $B$ is compact in $C$
    \end{enumerate}
\end{proposition}
\begin{proof}
    First, let us notice that since $B$ is substituted for $\OrbAs{A}{q}$ (that is a non re-entering component), then the states of $B$ constitute a non re-entering component in $C$.
    Since the set of out-transitions of $Q_B$ in $C$ is $\{(f_B, b, p) \mid f_B \in F_B \wedge \exists (o, b, p) \in \dout{A}{\OrbAs{}{q}}\}$, and since either $F_B$ is included in $F_C$ or $F_B$ and $F_C$ are disjoint sets, $Q_B$ is transverse in $C$.
    
    (\ref{subst_i1}): Since $A$ and $B$ are both $k$-block, then so is $C$.
    
    (\ref{subst_i2}): The in-transitions (out-transitions) of $Q_B$ in $C$, and the ones of $\OrbAs{}{q}$ in $A$ have the same labels.
    If $F_B$ and $F_C$ are disjoint sets, then $\FBA{C} = \FBA{A}$.
    Otherwise, $F_B$ is included in $F_C$ and since $\FBA{B}$ is a subset of $\FBA{\AutOrb{A}{q}}$, $\FBA{C}$ is a subset of $\FBA{A}$.
    
    (\ref{subst_i3}): Let us suppose that $C$ is not deterministic, then there exists a gate $g$ in $F_B$, an internal transition $(g, w, q)$ of $Q_B$ and an out-transition $(g, w', q')$ of $Q_B$ in $C$ such that $\{w, w'\}$ is not prefix-free.
    Thus, $w$ belongs to $\FBA{B}$, which is included in $\FBA{\AutOrb{A}{q}}$ and $w'$ is the label of an out-transition of $\OrbAs{}{q}$ in $A$.
    This contradicts the determinism of $A$.
    
    (\ref{subst_i4}): Since $B$ and $\AutOrb{A}{q}$ are $\mathcal{L}$-equivalent, the family of right languages of $B$ is equal to the set of internal languages of $\OrbAs{}{q}$ in $A$.
    Moreover, the external language of $Q_B$ in $C$ is the same as the one of $\OrbAs{}{q}$ in $A$.
    Thus, the set of right languages of the states of $B$ in $C$ is the same as the one of the states of $\OrbAs{}{q}$ in $A$.
    Then, by redirecting the in-transitions of $\OrbAs{}{q}$ to any state of $B$ equivalent to a state of $\AutOrb{A}{q}$, $A$ and $C$ are $\mathcal{L}$-equivalent.
    
    (\ref{subst_i5}): The orbital structure of $B$ is preserved in $C$.
    Since $A$ and $B$ have the orbit property, $Q_B$ is transverse in $C$, and the set of in-transitions of $Q_B$ in $C$ is $\{(p, b, h(o)) \mid (p, b, o) \in \din{A}{\OrbAs{}{q}}\}$, the automaton $C$ has the orbit property.
    Since both $A$ and $B$ pass the \textrm{BW}-test, thus $C$ also passes the \textrm{BW}-test.   
        
    (\ref{subst_i6}): $Q_B$ is transverse in $C$ and $C$ is deterministic.
    Thus, following Proposition~\ref{prop:lorb}, if $B$ is compact, then the states of $B$ in $C$ all have distinct right languages.
    Thus, every orbit of $B$ in $C$ is compact.
\end{proof}

    Thus, if we substitute a compact automaton for a non-compact orbit, the automaton we get has one less non-compact orbit than the original one.
    And if we do the same for every non-compact orbit, we can get an equivalent automaton with only compact orbits.

\subsection{Cutting and adding synchronizing transitions}\label{ss:CutAdd}

    In order to compact the extracted orbital automata while preserving the \textrm{BW}-test, we remove their synchronizing transitions to get smaller automata (see Figure~\ref{fig:ExAlgoC}), and put them back after compaction.
    This operation preserves the determinism and the maximal size of the blocks.
    Moreover, as a direct consequence of Theorem~\ref{th:BWtest} and of the definition of the \textrm{BW}-test, coherently adding or removing outgoing transitions for all the final states preserves the \textrm{BW}-test.

\begin{corollary}\label{coro:BWt_cut}
    Let $A$ be an automaton and $(b,s)$ be a pair in $S_A$.
    Then $A$ passes the \textrm{BW}-test if and only if $\AutCut{A}{(b,s)}$ does.
\end{corollary}

\begin{example}
    The automaton $C$ presented in Figure~\ref{fig:ExAlgoC} is obtained by removing the synchronizing transitions $(4, aab, 1)$ and $(4, aaa, 1)$ of $B$ (presented in Figure~\ref{fig:ExAlgoB}).
\end{example}

\begin{figure}[H]
    \centering
    \begin{tikzpicture}
		\node[state, initial] (1) {$1$};
		\node[state, above right of=1] (3) {$3$};
		\node[state, below right of=1] (3') {$3'$};
		\node[state, accepting, below right of=3] (4) {$4$};
		\path[->]
            (1) edge node {$aa$} (3)
            (1) edge [swap] node {$ab$} (3')
            (3) edge node {$a$} (4)
            (3') edge [swap] node {$a$} (4)			
		;
	\end{tikzpicture}
	\caption{The automaton $C$}
    \label{fig:ExAlgoC}
\end{figure}
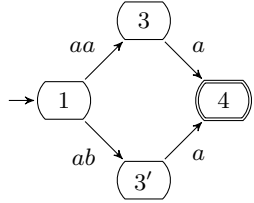

    Let us study the state-equivalence before and after cutting or adding synchronizing transitions.

	The \emph{$(b, s)$-cut language} of a state $p$, denoted by $\Lscut{}{p}{(b, s)}$, is the set of words which send $p$ to a final state without using any transition from $(F \times \{b\} \times \{s\})$.
	Since $(\delta \setminus (F \times \{b\} \times \{s\}))$ is a subset of $\delta$, we have $\Lscut{}{p}{(b, s)} \subset \LsA{}{p}$.
    If an automaton has a consistent label $b$ to a state $s$, the right language of a state can be expressed only with its $(b, s)$-cut language and the right language of $s$.
	
\begin{lemma}\label{lm:lcut}
	Let $A$ be an automaton such that $(b, s)$ belongs to $\cset{A}$.
	Then, for any state $p$ of $A$, $\LsA{}{p} = \Lscut{}{p}{(b, s)} \cdot (\{\varepsilon\} \cup \{b\} \cdot \LsA{}{s}) = \Lscut{}{p}{(b, s)} \cdot (\{b\} \cdot \Lscut{}{s}{(b, s)})^*$.
\end{lemma}
\begin{proof}
    By definition of the $(b, s)$-cut language, we have	$\LsA{}{p} = \Lscut{}{p}{(b, s)} \cup \Lscut{}{p}{(b, s)} \cdot \{b\} \cdot \LsA{}{s})$, which gives us $\LsA{}{s} = \Lscut{}{s}{(b, s)} \cup \Lscut{}{s}{(b, s)} \cdot \{b\} \cdot \LsA{}{s})$.
	From Arden Lemma~\cite{Ard61}, we can deduce that:
	\begin{align*}
		\LsA{}{s} &= (\Lscut{}{s}{(b, s)} \cdot \{b\})^* \cdot \Lscut{}{s}{(b, s)}\\
		&= \Lscut{}{s}{(b, s)} \cdot (\{b\} \cdot \Lscut{}{s}{(b, s)})^*
	\end{align*}
	Thus, for every state $p$ of $A$, we have:
	\begin{align*}
		\LsA{}{p} &= \Lscut{}{p}{(b, s)} \cdot (\{\varepsilon\} \cup \{b\} \cdot \LsA{}{s})\\
		&= \Lscut{}{p}{(b, s)} \cdot (\{\varepsilon\} \cup \{b\} \cdot \Lscut{}{s}{(b, s)} \cdot (\{b\} \cdot \Lscut{}{s}{(b, s)})^*)\\
		&= \Lscut{}{p}{(b, s)} \cdot (\{b\} \cdot \Lscut{}{s}{(b, s)})^*
	\end{align*}	
\end{proof}

\begin{proposition}\label{prop:lcut}
	Let $A$ be a deterministic block automaton such that $(b, s)$ belongs to $\cset{A}$.
	Then, for any two states $p$ and $q$ of $A$, we have $\LsA{}{p} = \LsA{}{q}$ if and only if $\Lscut{}{p}{(b, s)} = \Lscut{}{q}{(b, s)}$.
\end{proposition}
\begin{proof}
	Let us suppose that $\Lscut{}{p}{(b, s)} = \Lscut{}{q}{(b, s)}$.
	Since $(b, s) \in \cset{A}$, then, from Lemma~\ref{lm:lcut}, $\LsA{}{p} = \Lscut{}{p}{(b, s)} \cdot (\varepsilon + b \cdot \LsA{}{s}) = \Lscut{}{q}{(b, s)} \cdot (\varepsilon + b \cdot \LsA{}{s}) = \LsA{}{q}$.
	
	Let us suppose that $\LsA{}{p} = \LsA{}{q}$ and $\Lscut{}{p}{(b, s)} \neq \Lscut{}{q}{(b, s)}$.
	Then, there exists a shortest word $w$ in $\Lscut{}{p}{(b, s)}$ but not in $\Lscut{}{q}{(b, s)}$.
	Thus $w$ is in $\LsA{}{p} = \LsA{}{q}$.
	Then, from Lemma \ref{lm:lcut}, $w$ is in $\Lscut{}{q}{(b, s)} \cdot \{b\} \cdot \LsA{}{s}$.
	Consequently, there exists two words $u, v$ over $\Sigma$ such that $w = ubv$ and $u$ is in $\Lscut{}{q}{(b, s)}$.
	Since $|u| < |w|$ and $w$ is a shortest word, $u$ is in $\Lscut{}{p}{(b, s)}$.
	Since $w$ is in $\Lscut{}{p}{(b, s)}$, there exists a final state $f$ of $A$ such that $(p, u, f)$ belongs to $\delta^*$ and $(f, a, r)$ belongs to $\delta$, with $a$ a prefix of $bv$ such that if $a = b$ then $r \neq s$.
	However, since $(b, s)$ belongs to $\cset{A}$, $(f, b, s)$ belongs to $\delta$.
	Since either $a$ is a prefix of $b$ or $b$ is a prefix of $a$, this contradicts the determinism of $A$.
\end{proof}	

\begin{corollary}\label{coro:lcut_compact}
    Let $A$ be a deterministic block automaton such that $(b, s) \in \cset{A}$.
    Then, $\AutCut{A}{(b,s)}$ is compact if and only if $A$ is compact.
\end{corollary}

    Now, let us study the reverse of the $S$-cut operation.

    Let $A$ be a block automaton.
    A block $w$ is \emph{$A$-vacant} if there exists a state $q_w$ of $A$ such that every final state of $A$ has no outgoing transition labelled by $w$ to $q_w$.
	The \emph{set of vacancy of $A$} is denoted by $\vset{A} = \{(w, q_w) \mid \forall f \in F, (f, w, q_w) \notin \delta\}$.	
	Let $v = (w, q_w)$ be an element of $\vset{A}$. 
	The $v$\emph{-add} $\AutAdd{A}{v}$ of $A$ is constructed from $A$ by adding for each final state $f$ of $A$, the transition $(f, w, q_w)$.
	It is naturally extended to a subset of $\vset{A}$.

    Notice that if $(a,p)$ belongs to $\cset{A}$, then $\AutAdd{(\AutCut{A}{(a,p)})}{(a,p)} = A$, and symmetrically, if $(b,q)$ belongs to $\vset{A}$, then $\AutCut{(\AutAdd{A}{(b,q)})}{(b,q)} = A$.

\begin{example}
    The automaton $C'$ of Figure~\ref{fig:ExAlgoC'} is $\mathcal{L}$-equivalent to the automaton $C = \AutCut{B}{\{(aab,1), (aaa,1)\}}$ of Figure~\ref{fig:ExAlgoC}. 
    Then, the automaton $B'$ of Figure~\ref{fig:ExAlgoB'} is obtained from $C'$ by adding the synchronization transitions labelled by $aab$ and $aaa$ to the state $1$.
\end{example}

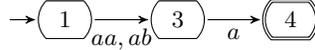
\begin{figure}[H]
    \centering
    \begin{tikzpicture}
		\node[state, initial] (1) {$1$};
		\node[state, right of=1] (3) {$3$};
		\node[state, accepting, right of=3] (4) {$4$};
		\path[->]
            (1) edge [swap] node {$aa, ab$} (3)
            (3) edge [swap] node {$a$} (4)
		;
	\end{tikzpicture}
	\caption{The automaton $C'$}
    \label{fig:ExAlgoC'}
\end{figure}

\begin{proposition}\label{prop:vadd_general}
    Let $A$ be a deterministic $k$-block automaton such that $(b,s)$ belongs to $\cset{A}$.
    Let $B$ be a deterministic $k$-block automaton $\mathcal{L}$-equivalent to $\AutCut{A}{(b, s)}$ such that $\FBA{B}$ is a subset of $\FBA{\AutCut{A}{(b,s)}}$, and $s'$ be a state of $B$ equivalent to $s$ in $\AutCut{A}{(b,s)}$.
    Then:
    \begin{enumerate}
        \item\label{vadd_i1} $\AutAdd{B}{(b,s')}$ is $k$-block
        \item\label{vadd_i2} $\FBA{\AutAdd{B}{(b,s)}}$ is a subset of $\FBA{A}$        
        \item\label{vadd_i3} $\AutAdd{B}{(b,s')}$ is deterministic
        \item\label{vadd_i4} $\AutAdd{B}{(b,s')}$ is $\mathcal{L}$-equivalent to $A$
        \item\label{vadd_i5} if $B$ passes the \textrm{BW}-test, then so does $\AutAdd{B}{(b,s')}$
        \item\label{vadd_i6} if $B$ is compact, then so is $\AutAdd{B}{(b,s')}$
    \end{enumerate}   
\end{proposition}
\begin{proof}
    (\ref{vadd_i1}): Since $A$ is $k$-block, $b$ is of length at most $k$.
    Moreover, since $B$ is $k$-block, $\AutAdd{B}{(b,s')}$ is also $k$-block.
    
    (\ref{vadd_i2}): Since $\FBA{\AutAdd{B}{(b,s)}} = \FBA{B} \cup \{b\}$, $\FBA{A} = \FBA{\AutCut{A}{(b,s)}} \cup \{b\}$, and $\FBA{B}$ is a subset of $\FBA{\AutCut{A}{(b,s)}}$, $\FBA{\AutAdd{B}{(b,s)}}$ is a subset of $\FBA{A}$.
    
    (\ref{vadd_i3}): Since $A$ is deterministic and $b$ is $A$-consistent, for any word $w$ of $\FBA{A} \setminus \{b\}$, the set $\{b, w\}$ is prefix-free.
    This still holds when $w$ is in $\FBA{B} \subset \FBA{\AutCut{A}{(b,s)}}$.
    Thus, $\AutAdd{B}{(b,s')}$ is deterministic.
    
    (\ref{vadd_i4}): Let $p$ be a state of $B$.
    Following Lemma~\ref{lm:lcut}, we have $\LsA{\AutAdd{B}{(b,s')}}{p} = \LsA{B}{p} \cdot (\{b\} \cdot \LsA{B}{s'})^*$.
    Since $B$ is $\mathcal{L}$-equivalent to $\AutCut{A}{(b,s)}$, there exists a state $p'$ in $\AutCut{A}{(b,s)}$ equivalent to $p$ in $B$.
    Thus, this state also belongs to $A$ and $\LsA{A}{p'} = \LsA{\AutCut{A}{(b,s)}}{p'} \cdot (\{b\} \cdot \LsA{\AutCut{A}{(b,s)}}{s})^*$.
    Since $s'$ in $B$ is equivalent to $s$ in $\AutCut{A}{(b,s)}$, we have $\LsA{A}{p'} = \LsA{B}{p} \cdot (\{b\} \cdot \LsA{B}{s'})^*$, and thus $\LsA{A}{p'} = \LsA{\AutAdd{B}{(b,s')}}{p}$.
    Moreover, for any state $q$ of $\AutCut{A}{(b,s)}$, there is a state $q'$ in $B$ equivalent to $q$ in $\AutCut{A}{(b,s)}$.
    Similarly, we can conclude that $\LsA{\AutAdd{B}{(b,s')}}{q'} = \LsA{A}{q}$.
    Consequently, $\AutAdd{B}{(b,s')}$ is $\mathcal{L}$-equivalent to $A$. 
    
    (\ref{vadd_i5}, \ref{vadd_i6})): Since $A$ is deterministic, $b$ does not belong to $\FBA{\AutCut{A}{(b,s)}}$, which means that $b$ does not belong to $\FBA{B}$ and $(b,q)$ belongs to $\vset{B}$.
    Thus $B = \AutCut{(\AutAdd{B}{(b,s')})}{(b,s')}$.
    Following Corollary~\ref{coro:BWt_cut}, if $B$ passes the \textrm{BW}-test, then so does $\AutAdd{B}{(b,s')}$.
    Finally, following (\ref{vadd_i3}) and Corollary~\ref{coro:lcut_compact}, if $B$ is compact, then so is $\AutAdd{B}{(b,s')}$. 
\end{proof}

\subsection{Eliminating unnecessary orbits}

	Several orbits may have the same $\Omega$-image.
	Thus, selecting only one representative orbit for each $\Omega$-image, we can compute an equivalent block automaton by linking them together.			    
    As any orbit admits a maximal orbit with the same $\Omega$-image, the selected representative ones are maximal.
    Moreover, if every orbit is compact, then the resulting automaton is compact.

\subsubsection{Completing maximal orbits}\
 
    The selected maximal orbits have to be linked together, but some of them may have out-transitions to a state which has no equivalent in any other selected orbit.
    Thus, these missing states have to be added as bridges between the selected orbits.

    A completion of a maximal orbit $O$ consists in adding states equivalent to some of $\CQ{O}$.
    To this end, we compute their transitions directly from $O$.
    Reading a word from a state may follow a path which ends up in the middle of some transitions.
    Thus, we define a function to complete the labels to reach their destination states.\\

    Let $A$ be a residual deterministic block automaton $A$. 
    The \emph{pathway} function $\Delta : Q \times \Sigma^* \rightarrow 2^{(\Sigma^* \times Q)}$ of $A$ is defined by $\Delta(p, u) = \{(v, q) \mid (p, uv, q) \in \delta^* \wedge \forall w \in \mathrm{Pref}(v) \setminus \{v\}, \delta^*(p, uw) = \emptyset\}$.
	In particular, $\Delta(p, w) = \{(\varepsilon, q)\}$ if and only if $(p, w, q)$ belongs to $\delta^*$.

\begin{example}
    Let us consider the state $4$ of the automaton $A'$ in Figure~\ref{fig:ExAlgoA'}: we have $\Delta(4, \varepsilon) = \{(\varepsilon, 4)\}$, $\Delta(4, a) = \{(aa, 1), (ab, 1), (ac, f_1), (b, f_1)\}$, $\Delta(4, aa) = \{(a, 1), (b, 1), (c, f_1)\}$, $\Delta(4, aaa) = \{(\varepsilon, 1)\}$ and $\Delta(4, aaaa) = \{(a, 3), (b, 3)\}$.
\end{example}
		
\begin{lemma}\label{lemma:langage_completion}
    Let $p$ be a state of $A$.
    If $(\Phi(p), w, s)$ belongs to $\delta_M^*$, then $\LsA{}{s} = \bigcup_{(u,r) \in \Delta(p,w)} u \cdot \LsA{}{r}$.
\end{lemma}
\begin{proof}
    Straightforward from the definition of $\Delta$.
\end{proof}

    Let $M$ be the minimal DFA of $A$.
    A completion of $O$ in $A$ with respect to $s$ in $\CQ{O}$ consists in selecting a state $o$ in $O$ and a word $w$ such that $(\Phi(o), w, s)$ belongs to $\delta_M^*$ to create an automaton $B = (\Sigma_A, Q, I_A, F_A, \delta)$ with $Q = Q_A \cup \{s_A\}$ and $\delta = \delta_A \cup \{(s_A, u, r) \mid (u, r) \in \Delta(o, w)\}$.

    Notice that two different choices of $o$ and $w$ may end up with different resulting automata.
    Moreover, $\{s_A\}$ is a trivial and since $O$ is maximal, $O \cup \{s_A\}$ is a non re-entering component.

\begin{lemma}\label{lm:state_eq_comp}
	The state $s_A$ of $B$ is equivalent to $s$.
\end{lemma}
\begin{proof}
    Let $p$ be a state of $O$ and $w$ be a word such that $(\Phi(p), w, s)$ belongs to $\delta_M^*$.
    Since $s$ belongs to $\CQ{O}$ and $O$ is maximal, there is no state in $A$ equivalent to $s$ that can be reached from a state of $O$.   
    Thus, $\Delta(p, w) \neq \{(\varepsilon, r)\}$ and since $s_A$ is not final, following Lemma~\ref{lemma:langage_completion}, $s_A$ is equivalent to $s$.
\end{proof}

   Conditions needed that are satisfied by the completion of a maximal orbit are as follow:
    
\begin{proposition}\label{prop:completion}
    Let $M$ be the minimal DFA of a residual deterministic $k$-block automaton $A$.
    Let $O$ be a maximal orbit of $A$ and $B$ be the completed of $O$ in $A$ with respect to $q$ such that $q \in \Phi(Q_A)$. 
    Then:
    \begin{enumerate}
        \item\label{compl_i1} $B$ is $\mathcal{L}$-equivalent to $A$
        \item\label{compl_i2} $\FBA{B} = \FBA{A}$        
		\item\label{compl_i3} $B$ is deterministic
        \item\label{compl_i4} $B$ is $k$-block
        \item\label{compl_i5} if $A$ passes the \textrm{BW}-test, then so does $B$
    \end{enumerate}
\end{proposition}
\begin{proof}
	(\ref{compl_i1}): Let $q_A$ be the state added in $B$. Following Lemma~\ref{lm:state_eq_comp}, $q_A$ is equivalent to $q$.
	Moreover, the right language of every state of $A$ is the same in $B$ and the initial state has not changed.
    Since $q$ belongs to $\Phi(Q_A)$, the families of right languages of $A$ and $B$ are the same.    

    (\ref{compl_i2}): Since $F_B = F_A$ and the structure of the states of $A$ is preserved, the set of final labels is the same for $A$ and $B$.

	(\ref{compl_i3}): Let $p$ be a state of $O$ and $w$ be a word of $\Sigma^*$ such that $(\Phi_A(p), w, q)$ belongs to $\delta_M^*$.
    Since $q$ belongs to $\CQ{O}$ and $O$ is maximal, there is no state in $A$ equivalent to $q$ that can be reached from a state of $O$.   
    Thus, $\Delta_A(p, w) \neq \{(\varepsilon, r)\}$.
              
    Since $A$ is deterministic, there exists exactly one state $s$ in $A$ and two words $w_p$, $w_s$ in $\Sigma^*$ such that $w = w_pw_s$, $(p, w_p, s)$ belongs to $\delta_A^*$ and for any $(u, q)$ in $\Delta_A(p, w)$, $(s, w_su, q)$ belongs to $\delta_A$.
    Thus, for any distinct $(u_1, q_1)$, $(u_2, q_2)$ in $\Delta_A(p, w)$, $(s, w_su_1, q_1)$ and $(s, w_su_2, q_2)$ belongs to $\delta_A$, and $w_su_1$ and $w_su_2$ are not prefix from each other.
    Thus $u_1$ and $u_2$ are not prefix from each other, which means that $B$ is deterministic.
    
    (\ref{compl_i4}): If $A$ is $k$-block, then for any state $p$ of $A$, for any word $w$ in $\Sigma^*$ and for any $(u, q)$ in $\Delta_A(p,w)$, we have $|u| < k$.
    
    (\ref{compl_i5}): Since there is no transition to $q_A$, it is a trivial orbit.
    As we do not change the orbital structure, if $A$ passes the \textrm{BW}-test, then so does $B$.
\end{proof}

\begin{example}
    Let the automaton $M$ in Figure~\ref{fig:ExAlgoM} be the minimal DFA of the block automaton $A'$ in Figure~\ref{fig:ExAlgoA'}.
    The orbit $O = \{1,3,4\}$ of $A'$ is maximal, $\CQ{O} = \{2, 5, 6\}$ and $\CQ{O} \cap \Phi(Q_{A'})= \{2, 5\}$.
    The block automaton in Figure~\ref{fig:ExAlgoCplA'} is a completion of $A'$ with respect to $2$ and $5$.
    
    The state $1$ (resp. $4$) in $A'$ is equivalent to the state $1$ (resp. $4$) in $M$, and $(1, a, 2)$ and $(4, a, 5)$ both belong to $\delta_M^*$.
    Thus, we compute the transitions of $2_{A'}$ and $5_{A'}$ from $\Delta_{A'}(1,a) = \{(a, 3), (b, 4)\}$ and $\Delta_{A'}(4,a) = \{(aa, 1), (ab, 1), (b, f_1), (ac, f_1)\}$.
\end{example}

\begin{figure}[H]
    \centering
    \begin{tikzpicture}
		\node[state, initial] (i) {$i_M$};
		\node[state, right of=i] (1) {$1$};
		\node[state, below of=1] (2) {$2$};
		\node[state, right of=2] (3) {$3$};
		\node[state, right of=3] (4) {$4$};
		\node[state, right of=4] (5) {$5$};
		\node[state, above of=5] (6) {$6$};
		\node[state, accepting, right of=6] (f) {$f$};
		\path[->]
			(i) edge node {$a, b$} (1)
			(1) edge [swap] node {$a$} (2)
			(2) edge [swap] node {$a, b$} (3)
			(3) edge [swap] node {$a$} (4)
			(4) edge [swap] node {$a$} (5)
			(5) edge [swap] node {$a$} (6)
			(5) edge [swap] node {$b$} (f)
			(6) edge [swap] node {$a, b$} (1)
			(6) edge node {$c$} (f)
			(f) edge [loop above] node {$a$} ()
		;
	\end{tikzpicture}
	\caption{The minimal DFA $M$ of $A'$}
    \label{fig:ExAlgoM}
\end{figure}
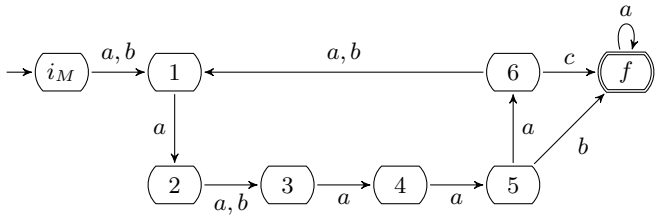

\begin{figure}[H]
    \centering
    \begin{tikzpicture}
		\node[state, initial] (i) {$i_A$};
		\node[state, above right of=i] (1) {$1$};
		\node[state, right of=1] (3) {$3$};
		\node[state, right of=3] (4) {$4$};
		\node[state, above of=4] (5A) {$5_{A'}$};
		\node[state, below of=3] (2A) {$2_{A'}$};
		\node[state, below right of=i] (5) {$5$};
		\node[state, below of=5] (2) {$2$};
		\node[state, right of=2] (3'') {$3''$};
		\node[state, accepting, right of=4, xshift=0.5cm] (f1) {$f_1$};
		\node[state, accepting, right of=f1] (f2) {$f_2$};		
		\path[->]
			(i) edge node {$a$} (1)
			(i) edge [swap, bend right=45] node {$ba$} (2)
			(1) edge [swap] node {$aa, ab$} (3)
            (3) edge [swap] node {$a$} (4)
            (4) edge [swap, bend right=25] node {$aab, aaa$} (1)
			(4) edge node {$ab, aac$} (f1)
			(2A) edge [swap] node {$a, b$} (3)
			(5A) edge [swap, bend right=25] node {$aa, ab$} (1)
			(5A) edge [bend left=25] node {$b, ac$} (f1)
			(2) [swap] edge node {$a, b$} (3'')
			(3'') edge node {$aa$} (5)
			(5) edge node {$aaa$} (2)
			(5) edge [swap] node {$ab$} (1)
			(5) edge [bend right=25] node {$b, ac$} (f1)
			(f1) edge [swap] node {$a$} (f2)
			(f2) edge [loop above] node {$a$} ()
		;
	\end{tikzpicture}
	\caption{A completion of the maximal orbit $\{1, 3, 4\}$ of $A'$}
    \label{fig:ExAlgoCplA'}
\end{figure}
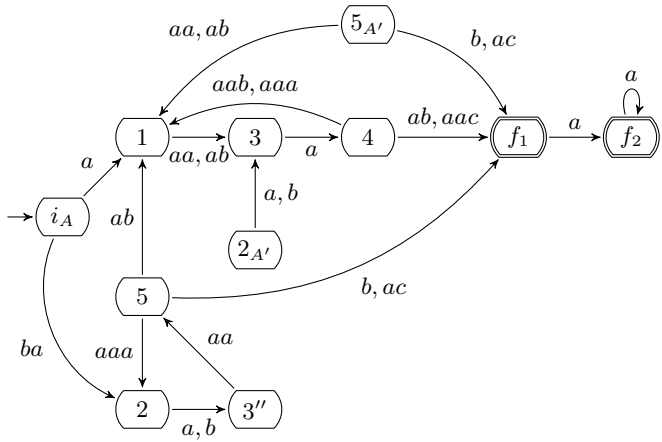

\subsubsection{The Slimming Procedure}\ 

    The slimming procedure consists in computing a $\mathcal{L}$-equivalent block automaton from a selection of maximal orbits.

    Let $M$ be the minimal DFA of a residual deterministic block automaton $A$.
    Let $(K_1, K_2, \ldots, K_l)$ be the set $\{\Omega(O) \mid O \text{ is an orbit of } A\}$, partially ordered with respect to reachability such that $K_i$ cannot reach $K_{i+j}$.
	Notice that $K_l$ contains the initial state of $M$ since $M$ is trim.
	
    The slimming procedure consists in the following steps:
    
\begin{enumerate}    
    \item Selecting $(O_1, O_2, \ldots, O_l)$ in $A$ such that $\Omega(O_j) = K_j$ and $O_j$ is maximal.
    \item Computing a block automaton $\textrm{Comp}(A)$ by completing every orbit $O_i$ of $A$ with the states of $Q_i = \CQ{O_i} \cap \Phi(Q_A)$.
        Then $C_i = O_i \cup Q_i$ is a \emph{completed orbit}.
    \item Computing a block automaton $\textrm{Link}(A)$ by linking every completed orbit of $\textrm{Comp}(A)$ together while preserving the transversality.
	\item Computing a block automaton $\textrm{Slim}(A)$ from $\textrm{Link}(A)$ by shifting the initial state to one in $C_l$ equivalent to $i_A$ and keeping only the states of the completed orbits.
\end{enumerate}
    
\begin{example}
    The automaton $A'$ in Figure~\ref{fig:ExAlgoA'} has only compact orbits.
    The automaton $M$ in Figure~\ref{fig:ExAlgoM} is its minimal DFA, and has three orbits: $K_1 = \{f\}$, $K_2 = \{1,2,3,4,5,6\}$ and $K_3 = \{i_M\}$.
    First, we select the following maximal orbits in $A'$: $O_1 = \{f_2\}$, $O_2 = \{1,3,4\}$ and $O_3 = \{i_A\}$.
    We have $\CQ{O_1} = \CQ{O_3} = \emptyset$ and $\CQ{O_2} \cap \Phi(Q_{A'}) = \{2, 5\}$.
    Thus, a completion of the orbits of $A'$ is presented in Figure~\ref{fig:ExAlgoCplA'}.
    Finally, an automaton $\textrm{Slim}(A')$ is presented in Figure~\ref{fig:ExAlgoCompA}.
\end{example}

\begin{figure}[H]
    \centering
    \begin{tikzpicture}
		\node[state, initial] (i) {$i_A$};
		\node[state, above right of=i] (1) {$1$};
		\node[state, right of=1] (3) {$3$};
		\node[state, right of=3] (4) {$4$};
		\node[state, above of=4] (5A) {$5_{A'}$};
		\node[state, below of=3] (2A) {$2_{A'}$};
		\node[state, accepting, right of=4, xshift=0.5cm] (f2) {$f_2$};
		\path[->]
			(i) edge node {$a$} (1)
			(i) edge [swap] node {$ba$} (2A)
			(1) edge [swap] node {$aa, ab$} (3)
            (3) edge [swap] node {$a$} (4)
            (4) edge [swap, bend right=25] node {$aab, aaa$} (1)
			(4) edge node {$ab, aac$} (f2)
			(2A) edge [swap] node {$a, b$} (3)
			(5A) edge [swap, bend right=25] node {$aa, ab$} (1)
			(5A) edge [bend left=25] node {$b, ac$} (f1)
			(f2) edge [loop above] node {$a$} ()
		;
	\end{tikzpicture}
	\caption{A slim of $A'$}
    \label{fig:ExAlgoCompA}
\end{figure}
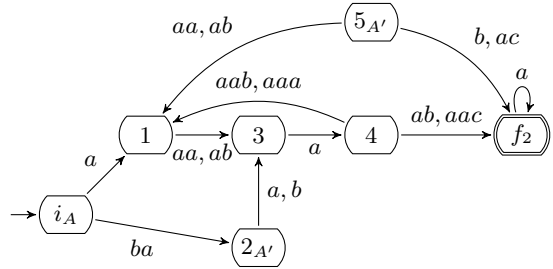

    Now, let us study the properties preserved by slimming a block automaton:   

\begin{proposition}\label{prop:slim_general}
	Let $A$ be a residual deterministic $k$-block automaton which passes the \textrm{BW}-test, and $B$ be a slim of $A$.
	Then
    \begin{enumerate}
        \item\label{slim_i1} $B$ is deterministic
        \item\label{slim_i2} $B$ is $k$-block
        \item\label{slim_i3} $B$ passes the \textrm{BW}-test
        \item\label{slim_i4} $B$ is $\mathcal{L}$-equivalent to $A$ and is residual
        \item\label{slim_i5} $\FBA{B}$ is a subset of $\FBA{A}$
        \item\label{slim_i6} if every orbit of $A$ is compact, then $B$ is compact.
    \end{enumerate}    	
\end{proposition}
\begin{proof}
    Let $\textrm{Comp}(A)$ and $\textrm{Link}(A)$ be the automata used to compute $B$.
    
    Following Proposition~\ref{prop:completion}, $\textrm{Comp}(A)$ is deterministic, $k$-block, $\mathcal{L}$-equivalent to $A$ and passes the \textrm{BW}-test.
    Moreover, the set of final labels is preserved, and since $A$ is residual, so is $\textrm{Comp}(A)$.

    Following the description of the completion of an orbit, every completed orbit constitute a non re-entering component.
    To link the completed orbits together, the out-transitions are redirected to an equivalent state.
    By construction, for any completed orbit $C_i$ in $\textrm{Comp}(A)$ and for any out-transition $(o, b, q)$ of $C_i$, there exists exactly one $C_{j<i}$ which contains at least one state equivalent to $q$.
    Thus the structure of every orbit of $\textrm{Comp}(A)$ and their transversality can be preserved in $\textrm{Link}(A)$.
    Since $\textrm{Comp}(A)$ passes the \textrm{BW}-test, all of its orbital automata pass the \textrm{BW}-test and it has the orbit property.
    Thus, $\textrm{Link}(A)$ also passes the \textrm{BW}-test.
    Moreover, redirecting the transitions to equivalent states preserves the right languages, the determinism, the $k$-block, and the set of final labels.    
    Thus, doing it one completed maximal orbit after the other, whatever the order is, we always get a $\mathcal{L}$-equivalent block automaton at each step.
    Since $\textrm{Comp}(A)$ is residual, so is $\textrm{Link}(A)$.
     
    Lastly, shifting the initial state makes the states that do not belong to a completed orbit non-accessible.
    Thus, removing these states preserve the determinism (\ref{slim_i1}), the $k$-block (\ref{slim_i2}), the \textrm{BW}-test (\ref{slim_i3}) and make $\FBA{B}$ a subset of $\FBA{\textrm{Link}(A)} = \FBA{A}$(\ref{slim_i5}).
    Moreover, the right language of any state in the completed orbits remains the same, and for any state of $\textrm{Link}(A)$, there exists an equivalent one in a completed orbit.
    Therefore $B$ is $\mathcal{L}$-equivalent to $\textrm{Link}(A)$ and thus to $A$, and since $\textrm{Link}(A)$ is residual, so is $B$ (\ref{slim_i4}).
         
    For any distinct selected orbits $O_i$ and $O_j$ in $A$, $\Omega(O_i) \neq \Omega(O_j)$.
    It means that any two states of $A$ belonging to distinct selected orbits cannot be equivalent.
    Thus, if every orbit of $A$ is compact, the selected orbits of $A$ are also compact in $B$. 
    Moreover, since every selected orbit are completed with states that have no equivalent in them, any two distinct states of $B$ are not equivalent (\ref{slim_i6}).
\end{proof}

\subsection{The compaction procedure}\label{ss:compaction}

    We define a procedure of compaction which preserves the determinism, the $k$-block and the \textrm{BW}-test.
    The resulting automaton is called \emph{a compacted} of the original one.
    
    Let $A$ be a residual, deterministic $k$-block automaton which passes the \textrm{BW}-test.
    We proceed as follows:
    
\begin{itemize}
    \item 
    if every orbit of $A$ is compact, a slim of $A$ is returned.

    \item 
    Otherwise, there exists a non-compact orbit $O$ of $A$ and a pair $(b,s)$ in the set of consistency of $O$.
    Then, the $(b, s)$-cut of the orbital automaton of $s$ is compacted itself.
    The previous synchronizing transition is put back on this compacted automaton to its initial state, to then substitute it for $O$ in $A$.
    Finally, the resulting automaton is compacted itself.
\end{itemize}

\begin{proposition}\label{prop:comp_general}
	Let $A$ be a residual deterministic $k$-block automaton which passes the \textrm{BW}-test, and $D$ be a compacted of $A$.
	Then
    \begin{enumerate}
        \item $D$ is deterministic
        \item $D$ is $k$-block
        \item $D$ passes the \textrm{BW}-test
        \item $D$ is $\mathcal{L}$-equivalent to $A$ and is residual
        \item $\FBA{D}$ is a subset of $\FBA{A}$
        \item $D$ is compact.
    \end{enumerate}    	
\end{proposition}
\begin{proof}
    If every orbit of $A$ is compact, then $D$ is a slim of $A$ and following Proposition~\ref{prop:slim_general}, $D$ is compact, $\mathcal{L}$-equivalent to $A$, residual, deterministic, $k$-block, passes the \textrm{BW}-test and its set of final labels is included in the one of $A$.

    Otherwise, there exists an orbit $O$ of $A$ which is not compact.
    Thus, it is not trivial and since $A$ passes the \textrm{BW}-test, the set of consistency of $O$ is not empty.
    Let $(b,s)$ be a pair in the set of consistency of $O$ and $B$ be the orbit automaton of $s$ in $A$.
    Since $B$ is an orbital automaton from $A$, it is trim, deterministic and $k$-block.
    Moreover, following Corollary~\ref{coro:BWt_orbAut}, since $A$ passes the \textrm{BW}-test, so does $B$.
    Let $C$ be the $(b,s)$-cut of $B$.
    Since $B$ is deterministic and $k$-block, so is $C$.
    Moreover, since $B$ is trim and has $s$ as its initial state, $C$ is also trim (and thus residual).
    Following Corollary~\ref{coro:BWt_cut}, since $B$ passes the \textrm{BW}-test, so does $C$.
    Let $C'$ be a compacted of $C$.
    Since $C$ is structurally smaller than $A$ (with at least one less transition), by induction on its size, $C'$ is compact, $\mathcal{L}$-equivalent to $C$, residual, deterministic, $k$-block, passes the \textrm{BW}-test, and its set of final labels is included in the one of $C$.
    Let $B'$ be the $(b, i_{C'})$-add of $C'$.
    Following Proposition~\ref{prop:vadd_general}, $B'$ is compact, $\mathcal{L}$-equivalent to $B$, residual, deterministic, $k$-block, passes the \textrm{BW}-test, and its set of final labels is included in the one of $B$.
    Let $A'$ be the substituted of $B'$ for $O$ in $A$.
    Following Proposition~\ref{prop:prop_subst}, $A'$ is $\mathcal{L}$-equivalent to $A$, residual, deterministic, $k$-block, passes the \textrm{BW}-test, and its set of final labels is included in the one of $A$.
    Moreover, $A'$ has one less non-trivial orbit than $A$.
    Then $D$ is a compacted of $A'$ and, by induction on the number of non-compact orbits of $A'$ compared to $A$, $D$ is compact, $\mathcal{L}$-equivalent to $A'$ (and thus to $A$), residual, deterministic, $k$-block, passes the \textrm{BW}-test and its set of final labels is included in the one of $A'$ (and thus of $A$).
\end{proof}

    Notice that the compaction procedure is based on the description of the \textrm{BW}-test, which always terminates by reaching acyclic automata through recursively eliminating transitions of orbital automata.
    In the case of the compaction procedure, this means reaching automata with only trivial orbits which are necessarily compact, implying the end of the recursive calls.
    Thus, the compaction procedure always halts.
    
\begin{example}
    Let us sum up our running example.
    The automaton $A$ (Figure~\ref{fig:ExAlgoA}) is $3$-block, deterministic and passes the \textrm{BW}-test.
    It is not compact and the orbit $O = \{1, 3, 3', 4\}$ is also not compact.
    The set of consistency of $O$ is $\{(aaa, 1), (aab, 1)\}$.
    Thus, we get the automaton $B = \AutOrb{A}{1}$ (Figure~\ref{fig:ExAlgoB}).
    Here, since the two synchronizing transitions go to the same state, both of them can be removed to get the automaton $C = \AutCut{B}{\{(aab,1), (aaa,1)\}}$ (Figure~\ref{fig:ExAlgoC}).
    Since it is acyclic, it has only compact orbits.
    Thus, $C'$ (Figure~\ref{fig:ExAlgoC'}) can be computed as a compacted and a slim of $C$.
    Then, $B'$ (Figure~\ref{fig:ExAlgoB'}) is obtained by putting back the synchronizing transitions of $B$, and $A'$ (Figure~\ref{fig:ExAlgoA'}) is the substitution of $B'$ for $\OrbAs{}{1}$ in $A$.
    Since $A'$ only has compact orbits, a compacted of $A$ can be computed (Figure~\ref{fig:ExAlgoCompA}).
\end{example}

    Following Theorem~\ref{th:KBD}, if a language is $k$-block deterministic, it is recognized by a deterministic $k$-block automaton which passes the \textrm{BW}-test.
    Since trimming an automaton preserves the \textrm{BW}-test, and since a trim deterministic block automaton is residual, by applying the compaction procedure, we can conclude that:

\begin{theorem}\label{thm:CKBD}
	A language is $k$-block deterministic if and only if it is recognized by a trim deterministic $k$-block automaton which is compact and passes the \textrm{BW}-test.
\end{theorem}  

    Thus, following Theorem~\ref{thm:KBD_FIN}, applying the \textrm{BW}-test over each one leads us to state that
    
\begin{theorem}
    The $k$-block determinism of a language is decidable.
\end{theorem}

\section{Conclusion and perspective}

    In this paper, we show that testing whether a language is $k$-block deterministic is decidable.
    First, we show that any $k$-block deterministic language is recognized by a compact deterministic $k$-block automaton which passes the \textrm{BW}-test.
    Furthermore, we show how to finitely generate all the compact deterministic $k$-block automata recognizing a given language and, as a direct consequence, that the $k$-block determinism is decidable.
    Notice that deciding whether there exists an integer $k$ such that a given regular language is $k$-block deterministic is still open.
    However, the generation of the successive $k$-transition automata define a semidecidable procedure.
    The next step of the study is to determine whether there exists an upper bound in order to decide when to halt this procedure.

\bibliography{biblio}
	
\end{document}